\documentclass[journal]{IEEEtran}

\usepackage{setspace}

\usepackage{graphics,
           psfrag,
           epsfig,
           amsthm,
           cite,
           amssymb,
           url,
           dsfont,
           subfigure,
           algorithm,
           algorithmic,
           balance,
           enumerate,
           color,
           setspace
}
\usepackage{amsmath}

\usepackage{epstopdf}

\newtheorem{theorem}{Theorem}

\newcommand{\eref}[1]{(\ref{#1})}
\newcommand{\sref}[1]{Section~\ref{#1}}
\newcommand{\appref}[1]{Appendix~\ref{#1}}
\newcommand{\fref}[1]{Figure~\ref{#1}}

\newcommand{\cref}[1]{Constraint~\ref{#1}}
\newcommand{\thref}[1]{Theorem~\ref{#1}}

\newcommand{\algref}[1]{Algorithm~\ref{#1}}

\newcommand{\ignore}[1]{}

\IEEEoverridecommandlockouts
\usepackage{lettrine}

\begin{document}

\title{\vspace{-.5cm}Resilient Backhaul Network Design Using Hybrid Radio/Free-Space Optical Technology}
\author{
\authorblockN{Ahmed Douik$^{\star}$, Hayssam Dahrouj$^\bullet$, Tareq Y. Al-Naffouri$^{\dagger\prime}$, and Mohamed-Slim Alouini$^\dagger$\\}%
\authorblockA{$^\star$California Institute of Technology (Caltech), California, United States of America\\
   $^\bullet$Effat University, Kingdom of Saudi Arabia \\
   $^\dagger$King Abdullah University of Science and Technology (KAUST), Kingdom of Saudi Arabia \\
    $^\prime$King Fahd University of Petroleum and Minerals (KFUPM), Kingdom of Saudi Arabia \\
   $^{\star}$\{ahmed.douik\}@caltech.edu,$^\bullet$hayssam.dahrouj@gmail.com,$^\dagger$\{tareq.alnaffouri,slim.alouini\}@kaust.edu.sa
    }
\vspace{-.8cm} }

\maketitle

\IEEEoverridecommandlockouts

\begin{abstract}
The radio-frequency (RF) technology is a scalable solution for the backhaul planning. However, its performance is limited in terms of data rate and latency. Free Space Optical (FSO) backhaul, on the other hand, offers a higher data rate but is sensitive to weather conditions. To combine the advantages of RF and FSO backhauls, this paper proposes a cost-efficient backhaul network using the hybrid RF/FSO technology. To ensure a resilient backhaul, the paper imposes a given degree of redundancy by connecting each node through $K$ link-disjoint paths so as to cope with potential link failures. Hence, the network planning problem considered in this paper is the one of minimizing the total deployment cost by choosing the appropriate link type, i.e., either hybrid RF/FSO or optical fiber (OF), between each couple of base-stations while guaranteeing $K$ link-disjoint connections, a data rate target, and a reliability threshold. The paper solves the problem using graph theory techniques. It reformulates the problem as a maximum weight clique problem in the planning graph, under a specified realistic assumption about the cost of OF and hybrid RF/FSO links. Simulation results show the cost of the different planning and suggest that the proposed heuristic solution has a close-to-optimal performance for a significant gain in computation complexity.
\end{abstract}


\section{Introduction}\label{sec:int}

\lettrine[lines=2]{W}{ith} the drastic increase of smartphones and data consuming devices, the last few years witnessed a gigantic increase in the demand for mobile data services, e.g., the demand is more than doubling each year for the last quinquennial. Moreover, mobile traffic is expected to experience $100$-fold increase by $2020$ \cite{546843}. To cope with such traffic growth, service providers are required to upgrade their networks substantially. For cells become simultaneously smaller and denser, a particular emphasis on the backhaul network upgrade is crucial.

Optical fibers (OF) are a popular high data-rate technology for the backhaul design supporting many Gbit/s, e.g., $9.9$ Gbit/s for STM-$64$ \cite{5473878}. However, their use requires digging and protection, which limits their application to specific scenarios excluding small cells. Furthermore, the employments of OF links entail high initial investment.

The radio-frequency (RF) technology, on the other hand, is a scalable and relatively easy to deploy solution for the backhaul planning. RF links are not limited by the geographic features of the location; however, the data rate they provide is lower to OF links rates. For the inverse relation between bandwidth and transmission range, RF technologies operating at high bandwidth, e.g., microwave and millimeter wave (mmwave), are limited in coverage. Further, due to spectrum shortage, the initial investment in the licensed spectrum is of high importance \cite{5185525}.

Recently, Free-Space Optics (FSO) technology emerge as a reasonable alternative for next-generation backhaul design. By transmitting a laser beam in the micrometer range \cite{64161469}, FSO photo-detector transceivers are immune to electromagnetic interference generated by nearby RF links. Such micrometer waves, also referred to as visible light, fall in the unlicensed part of the spectrum. Moreover, by using multiplexing techniques such as wavelength-division multiplexing, FSO can achieve up to $10$ Gbit/s over one kilometer and $1.28$ Tbit/s over $210$ meters \cite{64161469}. However, FSO performance is highly affected by weather conditions, e.g., rain, fog, and snow.

In order to benefit from both the low cost and reliability of the RF technology and the high data rate provided by the FSO technology, the hybrid RF/FSO technology is a suitable solution for backhaul design \cite{6844864}. Hybrid RF/FSO transceivers communicate using both the RF and FSO links and switch to FSO or RF only conditioned by the electromagnetic interference levels and the weather conditions.

The cost optimization problem of the backhaul is studied in \cite{545168,5484165} for different communication technologies. The authors in \cite{Smadi:09,Rajakumar:08} consider minimizing the cost of upgrading a network by optimally placing the minimum number of FSO transceivers to achieve a given target throughput. The hybrid RF/FSO is mainly considered in \cite{6134071,4609027,6777766}. The authors in \cite{6134071} propose a mixed integer program to upgrade an RF backhaul network by optimally placing FSO links. The setup is extended in \cite{4609027} to include reliability and throughput constraints. Reference \cite{6777766} suggests deploying mirrors for non-line-of-sight FSO connections. The backhaul design proposed in this paper is especially related to the works \cite{129456,646269,Dahrouj_backnet_magazine}. The authors in \cite{129456} propose designing a cost-efficient backhaul using hybrid RF/FSO under throughput constraints. The model is extended in \cite{646269} to incorporate reliability constraints. In \cite{Dahrouj_backnet_magazine}, the authors present a business case of an RF backhaul network, which motivates the need for deploying resilient hybrid RF/FSO backhauls.

This paper proposes the hybrid RF/FSO as a cost-effective technology for backhaul network design. To ensure a resilient backhaul, the paper imposes a given degree of redundancy by connecting each node through $K$ link-disjoint paths so as to cope with potential link failures. The paper then considers the network planning problem of minimizing the total deployment cost by choosing the appropriate link type, i.e., either hybrid RF/FSO or optical fiber, between each couple of base-stations by guaranteeing $K$ link-disjoint connections, a data rate target, and a reliability threshold.

The paper main contribution is to propose an explicit close-to-optimal solution to the backhaul planning problem under the aforementioned constraints, i.e., connectivity, reliability and data rate constraints. The paper proposes solving the problem using techniques from graph theory. It reformulates the problem as a maximum weight clique problem in the planning graph under a specified realistic assumption about the cost of OF and hybrid RF/FSO links. Simulation results suggest that the proposed heuristic solution provides a close-to-optimal performance for a significant gain in computation complexity.

\section{System Model and Parameters}\label{sec:sys}

This paper considers a network composed of $M$ base-stations denoted by the set $\mathcal{B}=\{b_1,\ \cdots,\ b_M\}$, wherein base-stations are connected to each other using either optical fibres or hybrid RF/FSO backhaul links. All nodes \footnote{The terms node and base-station are used interchangeably throughout this paper} are assumed to have a line-of-sight connection. The paper addresses the problem of minimizing the cost of backhaul planning under connectivity, reliability and data rates constraints and proposes choosing the appropriate cost-effective backhaul connection between BSs.

Let $\pi^{(O)}(b_i,b_j)=\pi^{(O)}_{ij}$ and $\pi^{(h)}(b_i,b_j)=\pi^{(h)}_{ij}$ be the cost of deploying an OF and a hybrid RF/FSO link, respectively, between nodes $b_i$ and $b_j$. As hybrid RF/FSO is a cost effective solution as compared to optical fibres, this paper assumes that $\pi^{(h)}_{ij} \leq \pi^{(O)}_{ij}$, $\forall (b_i,b_j) \in \mathcal{B}^2$.

The connectivity constraint is achieved by connecting each pair of nodes in the network via single or multi-hop connections through $1 \leq K < M$ link-disjoint paths. Such path diversity allows the network to be more resilient to link failure by providing multiple alternative routing solutions. \fref{fig:network} shows a planning for a network composed of $5$ base stations for $1$ and $2$ link-disjoint paths. Let $\mathcal{K}(b_i,b_j)=\mathcal{K}_{ij}$ be the path diversity between nodes $b_i$ and $b_j$, $\forall (b_i,b_j) \in \mathcal{B}^2$. In other words, $\mathcal{K}(b_i,b_j)$ is the maximum number of distinct and disjoint path that link node $b_i$ to $b_j$, e.g., $\mathcal{K}(1,4)=3$ and $\mathcal{K}(2,3)=2$ in \fref{fig:network} of $2$ link-disjoint path. Note that $2 \rightarrow 1 \rightarrow 4 \rightarrow 3$ and $2 \rightarrow 1 \rightarrow 5 \rightarrow 4 \rightarrow 3$ are not considered disjoint paths as they share the link $4 \rightarrow 3$.

\begin{figure}[t]
\centering
\includegraphics[width=0.8\linewidth]{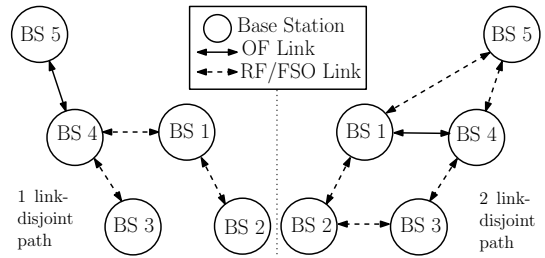}\\
\caption{Network containing $5$ base-stations connected together with OF and hybrid RF/FSO links for $1$ and $2$ link-disjoint paths.}\label{fig:network}
\end{figure}

Let $0 \leq \alpha \leq 1$ be the reliability threshold at each node. Unlike OF links that are perfectly reliable, the reliability of hybrid RF/FSO links heavily depends on several factors, such as transmission distance, power, and weather conditions. This paper assumes independent link failures. Let $R^{(O)}(b_i,b_j)=R^{(O)}_{ij}$ be the reliability of the link connecting nodes $b_i$ and $b_j$. Similarly, in contrast with OF links that always satisfy the targeted data rate $D_t$, the provided data of a hybrid RF/FSO link depends, for a fixed transmit power, on the distance separating the two end nodes and the weather conditions. Let $D^{(O)}(b_i,b_j)=D^{(O)}_{ij}$ be the provided data rate by the link between base-stations $b_i$ and $b_j$.

\section{Problem Formulation and Approximation} \label{sec:pro}

\subsection{Problem Formulation}

The paper proposes to minimize the deployment cost of the backhaul network by connecting base-stations using either optical fibers or hybrid RF/FSO links under the following practical constraint:
\begin{itemize}
\item C1: Each node can be connected to any other through either an OF or a hybrid RF/FSO link.
\item C2: Each node is connected to any other node with at least $K$ link-disjoint paths.
\item C3: The reliability threshold $\alpha$ is exceeded at each node.
\item C4: The provided data rate at each base-station is bigger or equal than the target data rate $D_t$.
\end{itemize}

Let $X_{ij}$ and $Y_{ij}, 1 \leq i,j \leq M$ be two binary variables indicating if base-stations $b_i$ and $b_j$ are connected through an OF connection, i.e., $X_{ij}=1$ or a hybrid RF/FSO, i.e., $Y_{ij}=1$, link, respectively. From the link reciprocity, the variables $X_{ij}$ and $Y_{ij}$ are symmetric. In other words, the variable should satisfy:
\begin{align}
X_{ij} &= X_{ji} \label{eq:c01} \\
Y_{ij} &= Y_{ji},\ 1 \leq i,j\leq M. \label{eq:c02}
\end{align}
Given that, at maximum, only one type of connection can exist between any pair of nodes, i.e., system constraint C1, it can easily be seen that such binary variables are mutually exclusive. In other words, the following condition is verified by any feasible solution:
\begin{align}
&X_{ij}Y_{ij} = 0,\ 1 \leq i,j\leq M. \label{eq:c1}
\end{align}

The connectivity constraint C2 can be reached by guaranteeing that the minimal maximum number of disjoint path, i.e., $\min_{1 \leq i \neq j \leq M} \mathcal{K}_{ij}$, between any couple of nodes exceeds $K$ link-disjoint paths. Hence, the connectivity constraint C2 can be written as follows:
\begin{align}
\min_{1 \leq i \neq j \leq M} \mathcal{K}_{ij} \geq K. \label{eq:c2}
\end{align}

Given independent link failures, the reliability condition C3 at node $b_i$ is violated if and only if all the links connecting base-station $b_i$ are experiencing a failure. Therefore, the system constraint C3 can be written as follows:
\begin{align}
&1 - \prod_{j=1}^M (1 - X_{ij})(1 - Y_{ij} R^{(h)}_{ij}) \geq \alpha,\ 1 \leq i\leq M. \label{eq:c3}
\end{align}

Finally, the data rate constraint C4 implies that the provided data rate at each node needs to exceed a predefined threshold $D_t$. Given that OF links always satisfy the target data rate, the offered data rate of such links can be assumed to be $D_t$ which allows to write the data rate constraint as follows:
\begin{align}
\sum_{j=1}^M X_{ij}D_t + Y_{ij}D^{(h)}_{ij} \geq D_t,\ 1 \leq i\leq M. \label{eq:c4}
\end{align}

Combining the constraints \eref{eq:c01}, \eref{eq:c02}, \eref{eq:c1}, \eref{eq:c2}, \eref{eq:c3}, and \eref{eq:c4}, the problem of minimizing the cost of the backhaul network planning can be formulated as:
\begin{subequations}
\label{Original_optimization_problem}
\begin{align}
\min & \ \sum_{i=1}^M \sum_{j=1}^M X_{ij}\pi^{(O)}_{ij} + Y_{ij}\pi^{(h)}_{ij} \\
{\rm s.t.\ } & X_{ij} = X_{ji} \label{eq:1} \\
&Y_{ij} = Y_{ji} \label{eq:2} \\
&X_{ij}Y_{ij} = 0 \label{eq:3} \\
&\min_{1 \leq i \neq j \leq M} \mathcal{K}_{ij} \geq K \label{eq:4} \\
&1 - \prod_{j=1}^M (1 - X_{ij})(1 - Y_{ij} R^{(h)}_{ij}) \geq \alpha \label{eq:5} \\
&\sum_{j=1}^M X_{ij}D_t + Y_{ij}D^{(h)}_{ij} \geq D_t\label{eq:6} \\
&X_{ij},Y_{ij} \in \{0,1\},\ 1 \leq i,j\leq M \label{eq:7},
\end{align}
\end{subequations}
where the optimization is over both binary variables $X_{ij}$ and $Y_{ij}$.

The $0-1$ mixed integer program proposed in \eref{Original_optimization_problem} is hard to solve and may involve a search over all possible combinations of the binary variables, which results in a high computational complexity. The difficulty lies, in particular, in the connectivity constraint \eref{eq:4} and the concurrent optimization over both binary variables $X_{ij}$ and $Y_{ij}$. Let the optimal solution to the optimization problem \eref{Original_optimization_problem} be called the optimal planning. To overcome such computation bottleneck, this paper proposes to approximate the optimization problem by a more tractable one. Hence, the next subsection suggests solving the problem for only the variables $X_{ij}$. In other words, it aims to discover the minimal cost planning solution when only optical fibers are allowed. Afterwards, such solution, referred to as the optical fibre planning, is used to reformulate the problem as a maximum weight clique problem under the assumption that long distance hybrid RF/FSO connections are more expensive than short distance OF links. Such assumption is justified by the fact that, for short distances, OF links whose cost mainly depends on the link's length, are cheaper than hybrid RF/FSO ones whose cost heavily depends on the transceivers price. The solution to the maximum weight clique problem is called the hybrid RF/FSO planning.

\subsection{Backhaul Design Using Optical Fiber Only}

This subsection considers the problem of backhaul network design using only optical fibres. The first part of this paragraph simplifies the problem formulation when only OF links are allowed. Afterwards, an algorithm to reach the optimal OF planning is proposed. By setting $Y_{ij} = 0,\ 1 \leq i,j\leq M$ in the original problem formulation \eref{Original_optimization_problem}, the OF planning problem can be written as follows:
\begin{subequations}
\label{eq:8}
\begin{align}
\min & \ \sum_{i=1}^M \sum_{j=1}^M X_{ij}\pi^{(O)}_{ij} \\
{\rm s.t.\ } & X_{ij} = X_{ji} \label{eq:9} \\
&\min_{1 \leq i \neq j \leq M} \mathcal{K}_{ij} \geq K \label{eq:10} \\
&1 - \prod_{j=1}^M (1 - X_{ij}) \geq \alpha \label{eq:11} \\
&\sum_{j=1}^M X_{ij}D_t \geq D_t\label{eq:12} \\
&X_{ij} \in \{0,1\},\ 1 \leq i,j\leq M \label{eq:13},
\end{align}
\end{subequations}

In order to simplify the problem formulation, constraints \eref{eq:11} and \eref{eq:12} are shown to be redundant in \eref{eq:8}. For any feasible solution $X_{ij}$, constraint \eref{eq:10} implies that $\mathcal{K}_{ij} \geq K,\ 1 \leq i,j\leq M$. In particular, as $K \geq 1$, we get $\mathcal{K}_{ij} \geq 1$, which implies:
\begin{align}
\sum_{j=1}^M X_{ij} \geq 1.
\label{eq:14}
\end{align}
Finally, using the inequality \eref{eq:14}, it becomes clear that constraints \eref{eq:11} and \eref{eq:12} are redundant. Therefore, the OF planning problem can be simplified to:
\begin{subequations}
\label{eq:15}
\begin{align}
\min& \ \sum_{i=1}^M \sum_{j=1}^M X_{ij} \pi^{(O)}_{ij} \\
{\rm s.t.\ }& X_{ij} = X_{ji} \label{eq:16} \\
&\min_{1 \leq i \neq j \leq M} \mathcal{K}_{ij} \geq K \label{eq:17} \\
&X_{ij} \in \{0,1\},\ 1 \leq i,j\leq M.
\end{align}
\end{subequations}

The key idea to solving the optimization problem \eref{eq:15} mentioned above is to generate a network with $K$ link-disjoint paths by first creating a system with $K-1$ link-disjoint paths. By prohibiting the already existing connections, the aim is to find the optimal set of links so as to produce a connected network. Combining both solutions results in a system with $K$ link-disjoint paths. Therefore, this subsection suggests successively generating systems whose minimal maximum number of disjoint paths, i.e., $\min_{1 \leq i \neq j \leq M} \mathcal{K}_{ij}$ increases by $1$ at each iteration. More specifically, the algorithm first generates the optimal planning for a number of link-disjoint paths $\min_{1 \leq i \neq j \leq M} \mathcal{K}_{ij}=1$. Afterwards, the algorithm adds connections to such optimal $1$ link-disjoint path solution to produce the optimal $2$ link-disjoint paths network. The process is repeated $K$ times so as to achieve the required resilience.

In order to generate the optimal $k+1$ link-disjoint paths connected network, given the optimal $k$ link-disjoint paths system, this paper proposes a clustering solution in which the cheapest links between any two clusters are successfully created; the groups are then merged. As explained above, the existing connections in the previous iterations of the algorithm are prohibited in the next iterations. This can easily be done by considering a modified cost function $\overline{\pi}^{(O)}(.)$ that takes the original value of the cost function $\pi^{(O)}(.)$ by the link have never been used before and $\infty$ otherwise. Therefore, to generate a $k+1$ link-disjoint paths connected network, a cluster is created for each base-station. The price of connecting two clusters is computed as the minimal cost of joining each couple of base-stations in the clusters. In other words, the cost of connecting the cluster $Z$ and $Z^\prime$ is defined as:
\begin{align}
\overline{\pi}^{(O)}(Z,Z^\prime) = \min\limits_{\substack{b \in Z \\ b^{\prime} \in Z^{\prime}}} \overline{\pi}^{(O)}(b , b^{\prime}).
\end{align}

The cheapest link between two clusters is deployed, and the cluster is merged into a single one. Such a process is repeated until all the initial clusters are merged into a single one. The steps of the algorithm are summarized in \algref{alg1} whose performance is characterized by the following theorem:
\begin{theorem}
\algref{alg1} produces the optimal solution to the optimization problem \eref{eq:15}.
\label{th1}
\end{theorem}

\begin{algorithm}[t]
\begin{algorithmic}
\REQUIRE $\mathcal{B}$, $K$, and $\pi^{(O)}(.)$.
\STATE Initialize $X_{ij}=0,\ 1 \leq i,j\leq M $.
\STATE Initialize $\overline{\pi}^{(O)}_{ij} = \pi^{(O)}_{ij},\ 1 \leq i,j\leq M $
\FOR{$k=1:K$}
\FOR{$i=1:M$}
\FOR{$j=i+1:M$}
\IF{$X_{ij} = 1$}
\STATE Set $\overline{\pi}^{(O)}_{ij}= \overline{\pi}^{(O)}_{ji} = \infty$.
\ENDIF
\ENDFOR
\ENDFOR
\STATE Initialize $\mathcal{Z} = \varnothing$.
\FORALL {$b \in \mathcal{B}$}
\STATE Set $\mathcal{Z} = \{\mathcal{Z},\{b\}\}$.
\ENDFOR
\WHILE {$|\mathcal{Z}| > 1$}
\STATE Set $(Z_i,Z_j) = \arg \min\limits_{\substack{Z,Z^{\prime} \in \mathcal{Z} \\ Z \neq Z^{\prime}}} \left[ \min\limits_{\substack{b \in Z \\ b^{\prime} \in Z^{\prime}}} \overline{\pi}^{(O)}(b , b^{\prime}) \right] $.
\STATE Set $(b_i,b_j) = \arg \min\limits_{\substack{b \in Z_i \\ b^{\prime} \in Z_j}} \overline{\pi}^{(O)}(b , b^{\prime})$.
\STATE Set $X_{ij} = X_{ji} = 1$.
\STATE Set $\mathcal{Z} = \mathcal{Z} \setminus \{Z_i\} \setminus \{Z_j\}$.
\STATE Set $\mathcal{Z} = \{\mathcal{Z},\{Z_i,Z_j\}\}$.
\ENDWHILE
\ENDFOR
\end{algorithmic}
\caption{Optimal optical fibres planning.}
\label{alg1}
\end{algorithm}

\begin{proof}
In order to establish the optimality of the solution reached by \algref{alg1}, an induction approach is used. Given the optimal solution to a $k-1$ link-disjoint connected network, the algorithm is demonstrated to output the optimal solution to a $k$ link-disjoint connected network. This is done by first showing that the solution is a feasible one and that any other solution results in a higher cost. The complete proof can be found in \appref{ap1}.
\end{proof}

Let $\overline{X}_{ij}$ be the optimal solution to the optical fiber planning problem reached by \algref{alg1}. The next subsection relates such solution to the original optimization problem \eref{Original_optimization_problem} and suggests approximating it by a more tractable problem.

\subsection{Problem Approximation}

This subsection describes the solution of the optimal OF planning $\overline{X}_{ij}$ of the original optimization problem \eref{Original_optimization_problem} and suggests approximating it with a more tractable one. The fundamental rationale of the approximation is to make use of the OF planning to produce a $K$ link-disjoint connected graph. In fact, it can easily be seen that since the planning $\overline{X}_{ij}$ is a $K$ link-disjoint one, then any planning $X_{ij}$ and $Y_{ij}$ verifying $X_{ij}+Y_{ij}=1$ if $\overline{X}_{ij}=1$ also produces a $K$ link-disjoint graph. Furthermore, the non-existence of an OF link in the optimal OF planning, i.e., $\overline{X}_{ij}=0$, does not add an extra constraint on the feasibility of the plan $X_{ij}$ and $Y_{ij}$. Therefore, the following constraint is a subset of the connectivity constraint \eref{eq:4}:
\begin{align}
(X_{ij} + Y_{ij})\overline{X}_{ij} = \overline{X}_{ij}. \label{eq:18}
\end{align}

Using the approximations provided in \eref{eq:18}, the backhaul network design problem using the hybrid RF/FSO technology can be approximated by the following problem:
\begin{subequations}
\label{Approximate_optimization_problem}
\begin{align}
\min& \ \sum_{i=1}^M \sum_{j=1}^M X_{ij}\pi^{(O)}_{ij}+Y_{ij}\pi^{(h)}_{ij} \label{eq:20} \\
{\rm s.t.\ } & X_{ij} = X_{ji} \label{eq:21} \\
&Y_{ij} = Y_{ji} \label{eq:22}\\
&X_{ij}Y_{ij} = 0 \label{eq:23} \\
&(X_{ij} + Y_{ij})\overline{X}_{ij} = \overline{X}_{ij} \label{eq:24} \\
&1 - \prod_{j=1}^M (1 - X_{ij})(1 - Y_{ij} R^{(h)}_{ij}) \geq \alpha \label{eq:25} \\
&\sum_{j=1}^M X_{ij}D_t + Y_{ij}D^{(h)}_{ij} \geq D_t \label{eq:26}\\
&X_{ij},Y_{ij} \in \{0,1\},\ 1 \leq i,j\leq M.
\end{align}
\end{subequations}

Let $X_{ij}$ and $Y_{ij}$ be a solution to \eref{Approximate_optimization_problem}. As the constraint \eref{eq:18} is strictly included in \eref{eq:4} then the solution is feasible to the original optimization problem \eref{Original_optimization_problem}.

\section{Backhaul Design Using Hybrid Radio/Free-Space Optical Technology} \label{sec:bac}

The problem approximation provided in \eref{Approximate_optimization_problem} is equivalent to the one illustrated in Lemma 3 of \cite{646269}. Hence, this section proposes a similar method to efficiently solving the problem through the design of the set of neighbours and the planning graph. Afterwards, under the assumption that long distance hybrid RF/FSO connections are more expensive than short distance OF links, the section reformulates the problem as a maximum weight clique problem in the planning graph.

\subsection{Set of Neighbours}

Let the set $\overline{\mathcal{N}}_i$ of neighbours of base-station $b_i$ be defined as follows:
\begin{align}
\mathcal{N}_i = \left\{b_j \in \mathcal{B} \text{ such that } \pi^{(O)}_{ij} \leq \max_{b_j \in \mathcal{B}}\overline{X}_{ik} \pi^{(O)}_{ik} \right\},
\label{eq:28}
\end{align}

Let $b_{i^*}$ be the node that can be connected to $b_i$ with the cheapest OF link, i.e., $b_{i^*} = \arg \min_{b_j \in \mathcal{B}} \pi^{(O)}_{ij}$. This paper assumes that hybrid RF/FSO between non-neighbouring nodes is more expensive than OF links between each node and its closest neighbour. In other words, the paper assumes that the following equation holds $\forall\ (b_i,b_j) \notin \mathcal{N}_j \times \mathcal{N}_i$:
\begin{align}
\pi^{(O)}_{ii^*}+\pi^{(O)}_{jj^*} \leq \pi^{(h)}_{ij}.\label{eq:27}
\end{align}

The following theorem characterizes the optimal solution $X^*_{ij}$ and $Y^*_{ij}$ of the optimization problem \eref{Approximate_optimization_problem}:
\begin{theorem}
The optimal solution to \eref{Approximate_optimization_problem} satisfies that all connections for an arbitrary node $b_i$ are inside its set of neighbours $\mathcal{N}_i$. In other words, $X^*_{i,j}+Y^*_{i,j}=1$ only if $(i,j) \in \mathcal{N}_j \times \mathcal{N}_i$.
\label{th2}
\end{theorem}

\begin{proof}
To show this theorem, all the planning decisions that violate the condition stated in the theorem are proved to be suboptimal. In fact, it can be seen from \eref{eq:27} that, if the set of neighbours allows to have a $K$ link-disjoint graph, then the connection outside the set of neighbours can be replaced as follows:
\begin{itemize}
\item If the link is a hybrid RF/FSO connection then it is cheaper to replace the link by $2$ OF links as suggested in \eref{eq:27}.
\item If the link is an OF connection, then it is more expensive than a hybrid RF/FSO one that can be replaced by connection inside the set of neighbours.
\end{itemize}
Therefore, the optimal solution to \eref{Approximate_optimization_problem} satisfies that all connections for an arbitrary node $b_i$ are inside its set of neighbours $\mathcal{N}_i$. The details of the proof are omitted as it mirrors the steps used in proving Theorem 2 in \cite{646269}.
\end{proof}

\subsection{Planning Graph and Proposed Backhaul Design}

The planning graph $\mathcal{G}(\mathcal{V},\mathcal{E})$ is a tool introduced in \cite{646269} to solving the planning problem using the hybrid RF/FSO technology for a $1$ link-disjoint connected graph. Given the optimal solution to the optical fibre planning and the definition of the minimal set of neighbours, such tool can be extended to solve the planning problem for a $K$ link-disjoint connected graph.

To generate the planning graph, a vertex is created for each combination of connections inside each cluster that satisfy the reliability and the data rate constraint. More specifically, let $\mathcal{C}_i$ be the set of such possible combinations for base-station $b_i$ defined as follows:
\begin{align}
\mathcal{C}_i = \{((X_{ij_1},Y_{ij_1}),\ \cdots,\ (X_{ij_{|\mathcal{N}_i|}},&Y_{ij_{|\mathcal{N}_i|}})), \text{ such that } \nonumber \\
\bigcup_{j \in \mathcal{N}_i}b_{j} &= \mathcal{N}_i \nonumber \\
X_{ij}Y_{ij} &= 0, \forall \ j \in \mathcal{N}_i \nonumber \\
(X_{ij}+Y_{ij})\overline{X}_{ij} &= \overline{X}_{ij}, \forall \ j \in \mathcal{N}_i \nonumber \\
1 - \prod_{j \in \mathcal{N}_i} (1 - X_{ij})(1 - Y_{ij} R^{(h)}_{ij}) &\geq \alpha \nonumber \\
\sum_{j \in \mathcal{N}_i} X_{ij}D_t + Y_{ij}D^{(h)}_{ij} &\geq D_t \}.
\end{align}

For each possible combination $\gamma \in \mathcal{C}_i, \ 1 \leq i \leq M$, a vertex $v_{ij}, 1 \leq j \leq |\mathcal{C}_i|$ is generated. The weight of each vertex is defined as half the total cost of the links represented by that vertex, i.e., the weight of $\gamma \in \mathcal{C}_i$ is:
\begin{align}
w(\gamma) = -\cfrac{1}{2}\sum_{j \in \mathcal{N}_i} X_{ij}\pi^{(O)}_{ij}+Y_{ij}\pi^{(h)}_{ij}.
\label{eq:29}
\end{align}

Two distinct vertices representing different nodes are connected if the connections they represent are symmetric. In other words, vertices $v_{ij}$ and $v_{kl}$ with $i \neq k$ are adjacent with an edge in $\mathcal{E}$ if and only if the connections they represent are compatible, i.e., $(X_{ik},Y_{ik}) = (X_{ki},Y_{ki})$ if $(b_i,b_k) \in (\mathcal{N}_k,\mathcal{N}_i)$.

Given the planning graph as constructed above and using the result of Theorem 3 in \cite{646269}, the optimal planning \eref{Approximate_optimization_problem} using the hybrid RF/FSO technology is equivalent to a maximum weight clique in the planning graph that can be solved with moderate complexity using efficient algorithms, e.g., \cite{16513519,6607889,13265492}.

\section{Simulation Results}\label{sec:sim}

\begin{figure}[t]
\centering
\includegraphics[width=0.8\linewidth]{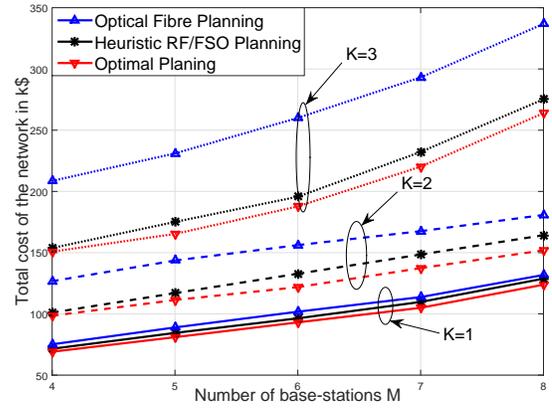}\\
\caption{Mean cost of the network versus the number of base-stations $M$. The solid lines are obtained for a $1$ link-disjoint network. The dashed and dotted lines are obtained for a $2$ and $3$ link-disjoint paths.}\label{fig:BC}
\end{figure}

\begin{figure}[t]
\centering
\includegraphics[width=0.8\linewidth]{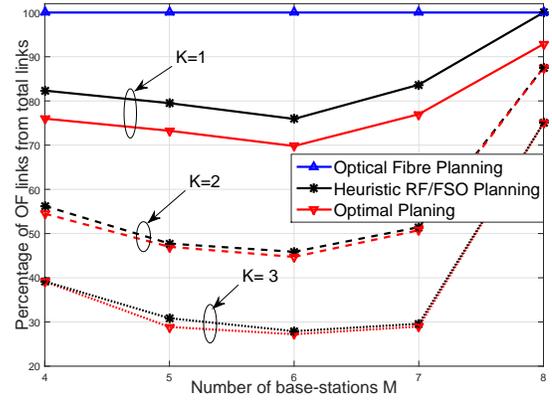}\\
\caption{Average percentage of used OF connections from the total number of used links versus the number of base-stations $M$. The solid lines are obtained for a $1$ link-disjoint network. The dashed and dotted lines are obtained for a $2$ and $3$ link-disjoint paths. Naturally, the three lines coincide for the OF planning.}\label{fig:BR}
\end{figure}

This section illustrates the performance of the proposed hybrid RF/FSO planning, i.e., optimal solution to \eref{Approximate_optimization_problem}, as compared with the optimal plan, i.e., optimal solution to \eref{Original_optimization_problem} and the OF planning, i.e., optimal solution to \eref{eq:15}, for different levels of network resilient. The system consists of a $5$ Km square area in which the base-stations are placed randomly at each iteration. As the price of the optical transceivers is negligible as compared to the cost of links, this paper does not consider its price. The price of a meter of the multi-mode OM$3$ $(50/125)$ OF is variable depending on the constructor, e.g., Asahi Kasei, Chromis, Eska, OFS HCS. This paper considers a medium price of $\pi^{(O)}=13.5\$$. On the other hand, as the amount of hybrid RF/FSO depends mainly on the cost of the transceiver, its price is assumed to be independent of the distance separating the two nodes and set to $\pi^{(h)}=20k\$$, a medium price according to different constructors, e.g., fSONA, LightPointe, and RedLine.

The reliability and provided data rate of a hybrid RF/FSO links are assumed to be only a function of the distance separating the two ends nodes. Furthermore, the paper considers that the reliability threshold $\alpha$ is satisfied for a distance of $2$ Km after which it decays exponentially. The provided data rate is considered to follows a similar model in which the targeted data rate is satisfied over $3$ Km and decreases exponentially for farther distances.

The reliability threshold is fixed to $\alpha=0.95$ in all the simulations. The number of base-stations $M$ and link-disjoint paths $K$ changes in the simulations so as to illustrate the performance of the proposed planning in different settings.

\begin{figure}[t]
\centering
\includegraphics[width=0.8\linewidth]{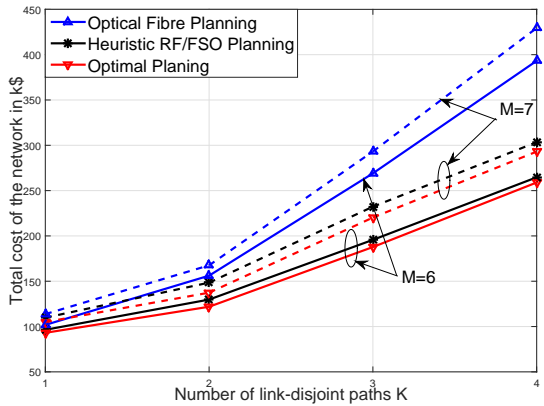}\\
\caption{Mean cost of the network versus the number of link-disjoint paths $K$. The solid lines are obtained for a system composed of $M=6$ base-stations. The dashed lines are obtained for a network with $M=7$ base-stations.}\label{fig:PC}
\end{figure}

\begin{figure}[t]
\centering
\includegraphics[width=0.8\linewidth]{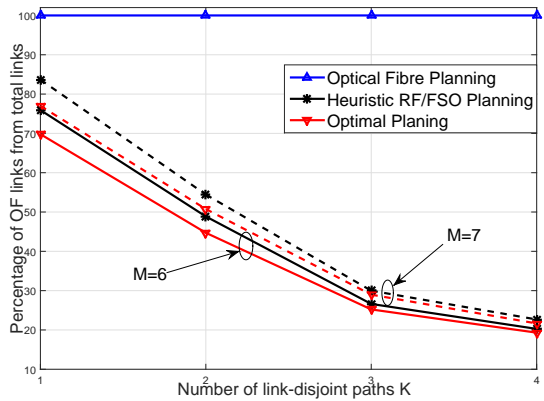}\\
\caption{Average percentage of used OF connections from the total number of used links versus the number of link-disjoint paths $K$. The solid lines are obtained for a network composed of $M=6$ base-stations. The dashed lines are obtained for a network with $M=7$ base-stations. Naturally, the two lines coincide for the OF planning.}\label{fig:PR}
\end{figure}

\fref{fig:BC} plots the total deployment cost of the network versus the number of base-stations $M$, for a various number of link-disjoint paths. It can easily be seen that the degradation of the proposed solution is negligible against the optimal solution for any number of base-stations and link-disjoint paths. The degradation is even less severe for a high number of base-stations with an enormous gain in complexity, especially for a large number of nodes in the network.

\fref{fig:BR} shows the ratio of the OF link used against the number of base-stations $M$, for a various number of link-disjoint paths. We can clearly see that for any number of base-stations, the percentage of used OF links decreases with the number of link-disjoint paths. This can be explained by the fact that the reliability and data rate constraint are satisfied for a small number of links. The additional links to guarantee the required resilience level in the network are hybrid FR/FSO links as they are more cost efficient that OF links.

Finally, to quantify the performance of the proposed algorithms with respect to the number of link-disjoint paths $K$, \fref{fig:PC} and \fref{fig:PR} show, respectively, the total cost of deployment and the percentage of used OF links versus the number of link-disjoint paths in networks composed of $M=6$ and $M=7$ base-stations. \fref{fig:PC} shows that both the optimal solution and the proposed one scales better with the number of link-disjoint paths as compared with the OF planning. This is due to the relative prices of hybrid RF/FSO and OF links. As the advantage of an OF link is especially its perfect reliability and satisfying data rate, for a densely connected network, these two constraints can be satisfied with a large number of hybrid RF/FSO links. As a result, the total cost of the system decreases drastically as compared with the OF planning as the number of link-disjoint paths increases.

The analysis is confirmed by \fref{fig:PR} that shows a constant decrease in the number of used OF links as the resilience degree increases in the network. Finally, it is worth mentioning that the performance of the proposed algorithm approaches the optimal solution for a large number of link-disjoint paths, which emphasises the close-to-optimal performance of the proposed algorithm.

\section{Conclusion}\label{sec:con}

This paper addresses the problem of resilient backhaul network design using the hybrid RF/FSO technology. In order to cope with link failure, the paper proposes a hybrid RF/FSO backhaul network which provides a minimum $K$ distinct and disjoint links connecting each two base-stations. The paper suggests finding the optimal type of links, OF or hybrid RF/FSO connection, between each two nodes so as to minimize the total deployment cost under the practical connectivity, reliability and data rate constraints. Given the complexity of the optimal solution, the paper approximates, under a specified realistic assumption about the cost of OF and hybrid RF/FSO links, the planning problem by reformulating it as a maximum weight clique in the planning graph, which can be solved using efficient algorithms. Simulation results suggest that the proposed heuristic solution has a close-to-optimal performance for a significant gain in computation complexity.

\section{Proof of \thref{th1}}\label{ap1}

This section shows that the solution reached by \algref{alg1} is the optimal solution to the optimization problem illustrated in \eref{eq:15}. To establish the result, an induction approach concerning the number of link-disjoint paths is used herein. In other words, this section aims to show that:
\begin{itemize}
\item The optimal solution to a $k$ link-disjoint path system contains the optimal solution to a $k-1$ link-disjoint path network.
\item Given the optimal solution to a $k-1$ link-disjoint path system, \algref{alg2} produces the optimal solution to a $k$ link-disjoint path network.
\item Each iteration $k$ of \algref{alg1} provides the same solution as \algref{alg2} with the optimal solution to a $k-1$ link-disjoint path system as an input.
\end{itemize}
Showing the steps displayed above is equivalent to showing that \algref{alg1}, for some iteration $K$, produces the optimal solution to the optimization problem \eref{eq:15}.

\begin{algorithm}[t]
\begin{algorithmic}
\REQUIRE $\mathcal{B}$, $X_{ij}^{k-1}$, and $\pi^{(O)}(.)$.
\STATE Initialize $X_{ij}^k=X_{ij}^{k-1},\ 1 \leq i,j\leq M $.
\STATE Initialize $\mathcal{Z} = \varnothing$.
\FORALL {$b \in \mathcal{B}$}
\STATE Set $\mathcal{Z} = \{\mathcal{Z},\{b\}\}$.
\ENDFOR
\WHILE {$|\mathcal{Z}| > 1$}
\STATE Set $(Z_i,Z_j) = \arg \min\limits_{\substack{Z,Z^{\prime} \in \mathcal{Z} \\ Z \neq Z^{\prime}}} \left[ \min\limits_{\substack{(b,b^{\prime}) \in Z \times Z^{\prime} \\ X_{b,b^{\prime}}^{k-1} \neq 1 }} \pi^{(O)}(d(b , b^{\prime})) \right] $.
\STATE Set $(b_i,b_j) = \arg \min\limits_{\substack{(b,b^{\prime}) \in Z \times Z^{\prime} \\ X_{b,b^{\prime}}^{k-1} \neq 1 }} \pi^{(O)}(d(b , b^{\prime}))$.
\STATE Set $X_{ij}^k = X_{ji}^k = 1$.
\STATE Set $\mathcal{Z} = \mathcal{Z} \setminus \{Z_i\}$.
\STATE Set $\mathcal{Z} = \mathcal{Z} \setminus \{Z_j\}$.
\STATE Set $\mathcal{Z} = \{\mathcal{Z},\{Z_i,Z_j\}\}$.
\ENDWHILE
\end{algorithmic}
\caption{Next optimal link-disjointed connected network.}
\label{alg2}
\end{algorithm}

Let the graph be abstracted in $\mathcal{G}(\mathcal{V},\mathcal{E})$ wherein $\mathcal{V}$ is the set of base-stations and $\mathcal{E}$ represents the set of edges, i.e., $\exists \ e_{ij} \in \mathcal{E}$ if and only if $X_{ij}=1$. Furthermore, let $X_{ij}^k$ be the optimal planning for a $k$ link-disjoint path network.

Showing that the optimal solution to the $k$ link-disjoint graph contains the optimal solution to a $k-1$ link-disjoint connected one translates into the following equation:
\begin{align}
X_{ij}^{k-1}X_{ij}^{k}=X_{ij}^{k-1}, \ 1 \leq i,j \leq M.
\label{eqap1}
\end{align}

Assume that $X_{ij}^{2}$ violates this property. By definition of a $2$ link-disjoint path, the removal of one edge from each vertex results in a connected graph, i.e., $1$ link-disjoint path. Therefore, let the edge be divided into two sets, $\mathcal{A}$ and $\mathcal{B}=\mathcal{E} \setminus \mathcal{A}$ such that $\mathcal{A}$ is the maximal set that can be removed from the graph, excluding edges in $X_{ij}^{1}$, resulting in a connected graph. In other words:
\begin{align}
\mathcal{A} = \arg \max_{A \in \mathbf{A}} |A|,
\label{eqap12}
\end{align}
where
\begin{align}
\mathbf{A} = \{A \in \mathcal{P}(\mathcal{E})\ | \ X_{ij}^1 \neq 1,&\forall \ e_{ij} \in A \text{ and } \label{eqap13} \\
& \qquad \mathcal{K}(\mathcal{G}(\mathcal{V},\mathcal{E}\setminus A))=1\}\nonumber,
\end{align}
with $\mathcal{P}(\mathcal{X})$ is the power set of $\mathcal{X}$ and $\mathcal{K}$ is the number of link-disjoint paths in the network defined in \sref{sec:sys}. Such edge separation is always possible as $\mathbf{A} \neq \varnothing$. Otherwise, the removal of $\{e_{ij} \in \mathcal{E} \ | \ X_{ij}^1=1\}$ results in a disconnected graph in contradiction with the assumption of a $2$ link-disjoint graph.

Let $\mathcal{B}^*$ be the set of edges obtained from $X_{ij}^1$. Given that $X_{ij}^2$ violates the property \eref{eqap1}, then $\mathcal{B}^* \nsubseteq \mathcal{B}$. Furthermore, as $\mathcal{B}^*$ is the optimal solution to a $1$ link-disjoint graph, the cost of $\mathcal{G}(\mathcal{V},\mathcal{B}^*\cup \mathcal{A})$ is lower than the one of $\mathcal{G}(\mathcal{V},\mathcal{E})$. Besides, given that $\mathcal{A}$ produces a $1$ link-disjoint graph, then $\mathcal{B}^*\cup \mathcal{A}$ is a feasible solution to a $2$ link-disjoint paths. Therefore, $X_{ij}^{2}$ is not the optimal solution which demonstrate that the optimal solution satisfy \eref{eqap1}.

Now assume that the property hold for a $k-1$ link-disjoint graph. A similar approach is used to show the property. By definition of a $k$ link-disjoint graph, the removal of one edge from each vertex results in a $k-1$ link-disjoint graph. Let the sets be divided in two set $\mathcal{A}_k$ and $\mathcal{B}_k=\mathcal{E} \setminus \mathcal{A}_k$ such that $\mathcal{A}_k$ is defined in a same manner as in \eref{eqap12} and $\mathbf{A}_k$ as in \eref{eqap13}, i.e.,
\begin{align}
\mathbf{A} = \{A \in \mathcal{P}(\mathcal{E})\ | \ X_{ij}^1 \neq 1,&\forall \ e_{ij} \in A \text{ and } \\
& \qquad \mathcal{K}(\mathcal{G}(\mathcal{V},\mathcal{E}\setminus A))=k-1\}\nonumber,
\end{align}

Using an argument similar to the one employed in the previous paragraph, it is easy to see that $\mathbf{A} \neq \varnothing$. Let $\mathcal{B}^*_k$ be the optimal solution to the planning with $k-1$ link-disjoint paths. Hence, as $X_{ij}^k$ violates the property and that $\mathcal{B}^*_k$ is the optimal solution, then $\mathcal{B}^*_k \nsubseteq \mathcal{B}_k$ and $\pi(\mathcal{G}(\mathcal{V},\mathcal{B}^*_k\cup \mathcal{A}_k)) \geq \pi(\mathcal{G}(\mathcal{V},\mathcal{E}))$. Finally, as $\mathcal{G}(\mathcal{V},\mathcal{B}^*_k\cup \mathcal{A}_k)$ is a $k$ link-disjoint graph, the property is shown to apply to all $k<M$ which concludes the proof.

From the analysis above, to produce a $k$ link-disjoint paths network, connections are added to a $k-1$ link-disjoint system. Furthermore, from the analysis above, removing the connections similar to the $k-1$ link-disjoint system, i.e., $\{e_{ij} \in \mathcal{E} \ | \ X_{ij}^{k-1}=1\}$, results in a connected network. Therefore, to produce the optimal solution to a $k$ link-disjoint paths network, the connections that should be added are those that produce a connected a network at minimum cost while prohibiting already existing connections. To solve the problem, this section designs \algref{alg2} as an extension of the Algorithm 1 proposed in \cite{117665}. The fundamental difference is that the newly designed algorithm prohibits connections existing in $X_{ij}^{k-1}$. As the proof follows similar steps than the one used in demonstrating Theorem 1 in \cite{646269}, they are omitted in this paper.

Finally, to conclude the proof it is sufficient to notice that prohibiting a given connection is similar to defining its weight as infinity. In fact, it is enough to show that there exists at least another link between any two arbitrary clusters $Z$ and $Z^\prime$ with a weight $\pi(Z,Z^\prime) < \infty$ to show the equivalence between \algref{alg1} and \algref{alg2}.

\bibliographystyle{IEEEtran}
\bibliography{citations}

\end{document}